\documentclass[onecolumn, a4size, 11pt]{IEEEtran}
\usepackage{float}
\usepackage{amsmath,amsfonts,amssymb}
\usepackage{graphicx}
\usepackage{cite}
\usepackage{multirow}
\usepackage{algorithm}
\usepackage{algorithmic}
\usepackage{subfigure}
\usepackage{multirow}
\usepackage{url}
\usepackage{color}
\usepackage{array}

\graphicspath{{./figs/}}

\IEEEoverridecommandlockouts

\makeatletter
\newcommand*{\rom}[1]{\expandafter\@slowromancap\romannumeral #1@}
\makeatother

\newtheorem{lemma}{Lemma}

\newcommand{\pr}{\mathrm{Pr}}

\newcommand{\xE}{\mathbb{E}}

\newcommand{\Qnet}{\bar Q_{\text{net}}}

\begin{document}
\title{Optimized Training Design for Multi-Antenna Wireless Energy Transfer in Frequency-Selective Channel}
%\title{Multi-Antenna Wireless Energy Transfer in Frequency-Selective Channels}
\author{\authorblockN{Yong~Zeng and Rui~Zhang\\}
\authorblockA{Department of Electrical and Computer Engineering, National University of Singapore\\
Email: \{elezeng, elezhang\}@nus.edu.sg}}
\maketitle

\begin{abstract}
This paper studies the optimal  training design for a multiple-input single-output (MISO) wireless energy transfer (WET) system in frequency-selective channels, where the frequency-diversity and energy-beamforming gains can be both achieved by properly learning the channel state information (CSI) at the energy transmitter (ET). By exploiting  channel reciprocity, a two-phase  channel training scheme is proposed to achieve the diversity and beamforming gains, respectively.  In the first phase, pilot signals are sent from the energy receiver (ER) over a selected subset of the available frequency sub-bands, through which the sub-band  that exhibits  the largest sum-power over all the antennas at the ET is determined and its index is sent back to the ER. In the second phase, the selected sub-band is further trained for the ET to estimate the multi-antenna channel and implement energy beamforming. We propose to maximize the \emph{net} energy harvested at the ER, which is the total harvested energy offset by that used for the two-phase channel training. The optimal training design, including the number of  sub-bands trained and the energy allocated for each of the two phases, is derived.
%The closed-form expression for the average harvested energy based on the proposed two-phase training scheme is derived, and the optimal training design is obtained to maximize the net harvested energy at the ER, which is the total energy harvested offset by that used for channel training.
\end{abstract}

\section{Introduction}
 Wireless energy transfer (WET) has drawn significant interests recently due to its great potential to provide cost-effective and reliable power supplies for energy-constrained wireless networks \cite{502}. One enabling technique of WET for long-range applications (say up to tens of meters) is via radio-frequency (RF) or microwave prorogation, where dedicated energy-bearing signals are transmitted from the energy transmitter (ET) for the energy receiver (ER) to harvest the RF energy (see e.g. \cite{525} and references therein). To overcome the significant power attenuation over distance, employing multiple antennas at the ET and advanced beamforming techniques to efficiently direct wireless energy to the destined  ER, termed {\it energy beamforming}, is an essential technique for WET \cite{478}. Similar to the emerging massive multiple-input multiple-output (MIMO) enabled wireless communications (see e.g. \cite{497} and references therein), by equipping a very large number of antennas at the ET, enormous  energy beamforming gain can be achieved; hence, the end-to-end energy transfer efficiency can be greatly enhanced.

  %One practical design issue for WET systems is to determine the frequency band on which the energy-bearing signals are transmitted. As no dedicated spectrum has been assigned for WET, one possibility is to operate on the license-free industry-science-medical (ISM) bands (e.g., around 0.9 or 2.4 GHz) \cite{}. Recently, with the concept known as simultaneous wireless information and power transfer (SWIPT), it has been found that WET can essentially co-exist with wireless communication systems, since the same transmitting waveforms can be designed to have dual purposes of both information and energy transmission. WET can also be implemented based on the cognitive radio concept \cite{Lee, Lu}, where energy signals are transmitted opportunistically over the spectrum that results the minimum or no interference to primary users. For any spectrum access scheme mentioned above, frequency-selective channel will be resulted if the total available spectrum exceeds the channel coherence bandwidth \cite{}. In this case, the frequency-diversity gain can be exploited to enhance the energy transfer efficiency, since intuitively, the signals should be transmitted over the sub-band that gives the highest channel gain for maximum energy delivery.

  On the other hand, for MIMO WET in a wide-band regime over frequency-selective channels, the frequency-diversity gain can also be exploited to further enhance the energy transfer efficiency, by transmitting more power over the sub-band with higher channel gain. WET in single-antenna or single-input single-output (SISO) frequency-selective channels has been studied in \cite{504,522,526} under the more general setup of simultaneous wireless information and power transfer (SWIPT), where perfect channel state information (CSI) is assumed at the transmitter. %For multi-antenna WET in frequency-selective channels, both energy-beamforming and frequency-diversity gains can be exploited if the channel state information (CSI) is properly learned at the ET.

  %All these works assume that the ET has the perfect knowledge of the channel state information (CSI).

  %For multi-antenna WET in frequency-selective channels, both energy-beamforming and frequency-diversity gains can be achieved if the channel state information (CSI) is properly learned at the ET.

   %In practice, CSI needs to be practically acquired by the ET at the cost of

   In practice, both the energy-beamforming and frequency-diversity gains in MIMO WET over frequency-selective channels can be achieved, but crucially depend on the available CSI at the ET, which needs to be practically obtained  at the cost of additional time and energy consumed. Similar to wireless communication,  a direct approach to obtain CSI is by sending pilot signals from the ET to the ERs, each of which estimates the corresponding channel and then sends the estimated channel back to the ET via a feedback channel  \cite{484,495}. However, since the training overhead increases with the number of antennas $M$ at the ET, this method is not suitable when $M$ is large. In \cite{491}, a new channel-learning design to cater for the practical RF energy harvesting circuitry at the ER has been proposed. However, the training overhead still increases quickly with $M$, and can be prohibitive for large $M$. In \cite{528}, by exploiting channel reciprocity between the forward (from the ET to the ER) and reverse (from the ER to the ET) links, we have proposed an alternative channel-learning scheme for WET based on the reverse-link training, which is more efficient since the training overhead becomes independent of $M$. However, the proposed design in \cite{528} applies only for narrowband flat-fading channels instead of the more complex broadband frequency-selective fading channels, which  motivates this work.

  In this paper, we consider a MISO point-to-point WET system over frequency-selective fading channels. To exploit both the frequency-diversity and energy-beamforming gains, we propose a two-phase channel training scheme by exploiting the channel reciprocity. In the first phase, pilot signals are sent from the ER over a selected subset of the available frequency sub-bands, each over an independent flat-fading channel. Based on the received total energy over all the antennas at the ET over each of the trained sub-bands, the ET determines the sub-band  that has the largest energy and sends its  index to the ER. In the second phase, the selected sub-band is further trained by the ER, so that the ET obtains an estimate of the exact MISO channel over this sub-band to implement energy beamforming. Due to  the limited energy harvested  at the ER, the training design needs to achieve a good balance between exploiting the diversity versus beamforming gains, yet without consuming excessive energy at the ER. Therefore, we propose to maximize the \emph{net} energy harvested at the ER, which is the total harvested energy offset by that used for both phases of  channel training. The optimal training design, including the number of sub-bands trained and the energy allocated for each of the two training phases, is derived. Simulation results are provided to validate our analysis.

 \section{System Model}\label{sec:systemModel}
We consider the MISO point-to-point WET system in frequency-selective channel, where an ET with $M\geq 1$ antennas is employed to deliver wireless energy to a single-antenna ER. We assume that the total bandwidth is $B$ Hz, which is equally divided into $N$ orthogonal sub-bands with the $n$th sub-band centered at  frequency $f_n$ and of bandwidth $B_s=B/N$. We assume that $B_s\ll B_c$, where $B_c$ denotes the channel coherence bandwidth, so that the channel between the ET and ER experiences frequency flat-fading within each sub-band. Denote $\mathbf h_n\in \mathbb{C}^{M\times 1}$, $n=1,\cdots, N$, as the baseband equivalent MISO channel from the ET to the ER in the $n$th sub-band. We assume a quasi-static  Rayleigh fading model, where $\mathbf h_n$ remains constant within each block of $T\ll T_c$ seconds, with $T_c$ denoting the channel coherence time, but can vary from one block to another. % based on Rayleigh fading. %. We further Rayleigh fading in this paper, which is a suitable model for rich scattering environment; while other channel models (e.g., Rician fading) will be considered in our future work. In this case,
 Furthermore, the elements in $\mathbf h_n$ are modeled as independent and identically distributed (i.i.d.) zero-mean circularly symmetric complex Gaussian (CSCG) random variables with  variance $\beta$, i.e.,
\begin{align}
\mathbf h_n \sim \mathcal{CN}(\mathbf 0, \beta \mathbf I_M), \ n=1,\cdots, N, \label{eq:hn}
\end{align}
where $\beta$ models the large-scale fading  due to shadowing as well as the distance-dependent path loss. %We further assume that %different sub-bands have independent channel fading, i.e.,
%$\mathbf h_n$ and $\mathbf h_m$ are independent, $\forall n\neq m$. %\footnote{If two sub-bands are completely correlated, then there is no loss of optimality to use only one of them for maximum energy transfer.}

\begin{figure}
\centering
\includegraphics[scale=0.8]{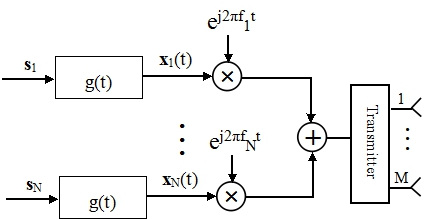}
\caption{Schematics of a  multi-antenna multi-band wireless energy transmitter.}\label{F:transmitter}
\end{figure}

Within each block of $T$ seconds, i.e., $0\leq t\leq T$, the input-output relation for the forward link energy transmission can be expressed as
\begin{align}
y_n(t)= \mathbf h_n^H \mathbf x_n(t) + z_n(t), \ n=1,\cdots, N, \label{eq:yt}
\end{align}
where $y_n(t)$ denotes the received signal at the ER; $\mathbf x_n(t)\in \mathbb{C}^{M\times 1}$ denotes the baseband energy-bearing signals transmitted by the ET in the $n$th sub-band; and $z_n(t)$ denotes the additive noise at the ER. Different from wireless communication where random signals need to be transmitted to convey information, $\mathbf x_n(t)$ in \eqref{eq:yt} is designated  only for energy transmission and thus can be chosen to be  deterministic. Denote by $P_f$ the total transmit power constraint at the ET over the $N$ sub-bands. We thus  have
%$\sum_{n=1}^N \| \mathbf x_n\|^2\leq P_f$.
\begin{align}\label{eq:Pf}
\frac{1}{T}\sum_{n=1}^N \int_0^T \left \| \mathbf x_n(t)\right \|^2 dt\leq P_f.
\end{align}

At the ER, the incident RF power captured by the antenna is converted to usable direct current (DC) power by a device called rectifier \cite{514}. %Due to the law of energy conservation, the total harvested RF-band energy over all the $N$ sub-bands during each coherent block, denoted by $Q$,   is proportional to the energy of the received baseband signal \cite{478}, i.e.,
By ignoring the energy harvested from the background noise which is practically small, the total harvested energy over all $N$ sub-bands during one  block can be expressed as \cite{478}
\begin{equation} \label{eq:Q}
 \begin{aligned}
  Q=\eta \sum_{n=1}^N \int_0^T \left|\mathbf h_n^H \mathbf x_n(t)\right|^2dt,
  \end{aligned}
 \end{equation}
 where $0<\eta\leq 1$ denotes the energy harvesting efficiency at the ER. %, and $\mathbf X_n$ is a positive-semidefinite matrix defined as $\mathbf X_n\triangleq \frac{1}{T} \int_0^T \mathbf x_n(t) \mathbf x_n^H(t)dt$. Without loss of optimality, we can assume that $\mathrm{rank}\left(\mathbf X_n\right)=1$, so that
  Without loss of generality, $\mathbf x_n(t)$ can be expressed as (see Fig.~\ref{F:transmitter} for the transmitter schematics)
 \begin{align}
 \mathbf x_n(t)=\mathbf s_n  g(t), \ 0\leq t\leq T, \ n=1,\cdots, N, \label{eq:xn}
 \end{align}
 where $\mathbf s_n \in \mathbb{C}^{M\times 1}$, and $g(t)$ represents the pulse-shaping waveform (e.g., raised cosine pulse) with normalized power, i.e., $\frac{1}{T} \int_0^T |g(t)|^2dt=1$. %The transmitter design for multi-band multi-antenna WET system is illustrated in Fig.~\ref{F:transmitter}.
   Note that the bandwidth of $g(t)$, which is approximately equal to $1/T$, needs to be no larger than $B_s$. We thus have
 \begin{align}
 %\frac{1}{T_c}<B_g\approx \frac{1}{T}<B_s<B_c, \label{eq:underspread}
 \frac{1}{T_c}\ll \frac{1}{T}<B_s \ll B_c, \label{eq:underspread}
 \end{align}
 or $T_cB_c\gg 1$, i.e., a so-called ``under-spread'' wide-band fading channel is assumed.

 %\approx T_d$, with $T_d$ denoting the delay spread of the channel \cite{}.  %i.e., a so-called under-spread fading scenario is assumed.
  %For instance, in indoor environment with rich scattering and slow mobility, $T_c$ can be on the order of milliseconds whereas $T_d$ is usually in micro-seconds, hence \eqref{eq:underspread} is satisfied.

  %Note that if the block length $T$ is sufficiently large, we may have $1/T\ll B_s$. In this case, only a very small portion of the spectrum in each sub-band is actually needed for energy transmission.
  %In the extreme case, wireless energy can be transmitted using unmodulated tones so that essentially zero spectrum is occupied \cite{}. However, such an operation is usually undesired, since, for one thing, the block duration $T$ is practically limited by the channel coherence time $T_c$ if the channel state information (CSI) needs to be exploited for energy transmission; for another, focusing all power in a number of frequency tones could easily violate the spectral mask constraints \cite{}. Therefore, in general, WET requires non-zero bandwidth \cite{}.

  From  \eqref{eq:xn}, the  power constraint in \eqref{eq:Pf} can be rewritten as $\sum_{n=1}^N \|\mathbf s_n\|^2\leq P_f$, and the harvested energy $Q$ in \eqref{eq:Q} can be expressed as $Q=\eta T \sum_{n=1}^N |\mathbf h_n^H \mathbf s_n|^2$. In the ideal case with perfect CSI, $\{\mathbf h_n\}_{n=1}^N$, at the ET, the optimal design of $\{\mathbf s_n\}_{n=1}^N$  that maximizes $Q$ can be obtained by solving the following problem
\begin{equation}\label{P:perfectCSI}
\begin{aligned}
\max \ & \eta T \sum_{n=1}^N \left|\mathbf h_n^H \mathbf s_n \right|^2 \\
\text{ subject to } & \sum_{n=1}^N \| \mathbf s_n\|^2\leq P_f.
\end{aligned}
\end{equation}

It can be easily shown  that the optimal solution to problem \eqref{P:perfectCSI} is
\begin{align}\label{eq:optxn}
\mathbf s_n= \begin{cases}
\sqrt {P_f} \frac{\mathbf h_n}{\|\mathbf h_n\|}, & \text{ if } n=\arg \underset{n^\prime=1,\cdots, N}{\max}\|\mathbf h_{n^\prime}\|^2, \\
\mathbf 0, & \text{ otherwise}.
\end{cases}
\end{align}
%where $s_n$ denotes arbitrary unit-power symbol.
The resulting harvested energy can be expressed as
\begin{align}
Q_{\max}=\eta T P_f \underset{n=1,\cdots, N}{\max} \left \| \mathbf h_n \right \|^2. \label{eq:Qmax}
\end{align}
It is observed from \eqref{eq:optxn} that for a MISO multi-band WET system with the sum-power constraint, the optimal energy transmission scheme  allocates all the available power to the sub-band with the largest MISO channel power. As a result, all the other sub-bands can be used for other applications such as communication.  %This is in contrast to multi-band wireless communications, for which in general more than one sub-bands should be allocated with non-zero power for maximum data rate transmission \cite{209}. %\footnote{The fundamental reason for such a difference is that while data rate is a strictly concave function of the transmission power, the amount of energy harvested at the ER is linear with respect to that transmitted by the ET.}
 The solution given in \eqref{eq:optxn} also indicates  that for the selected sub-band, maximum ratio transmission (MRT) should be performed across different transmit antennas at the ET  to achieve the maximum energy beamforming gain. %Equation \eqref{eq:Qmax} implies that for multi-antenna frequency-selective WET systems, both frequency-diversity gain and energy beamforming gain can be attained in the ideal case with perfect CSI at the ET.

In practice, the CSI $\{\mathbf h_n\}_{n=1}^N$ needs to be estimated  at the ET. By exploiting channel reciprocity, we propose a two-phase channel training scheme, as illustrated in Fig.~\ref{F:twoPhaseTraining}. The first  phase corresponds to the first $\tau_1< T$ seconds of each block, where pilot signals are sent by the ER to the ET over $N_1$ out of the $N$ available sub-bands, each with energy $E_1$. By estimating the received energy over all $M$ antennas at the ET over each of the $N_1$ trained sub-bands (whose indices are assumed to be known at the ET),  the ET determines the sub-band with  the largest power gain $\|\mathbf h_{n^{\star}}\|^2$, and sends the index $n^{\star}$ to the ER. In the second phase of $\tau_2< T-\tau_1$ seconds, additional training signal is sent by the ER in sub-band $n^{\star}$ with  energy $E_2$. %Based on the received training signal,
 The ET then obtains an estimate of the exact MISO channel $\mathbf h_{n^\star}$, based on which MRT-based energy beamforming is applied  during the remaining $T-\tau_1-\tau_2$ seconds of each block. The proposed two-phase training scheme is elaborated in more details in the next section.

\begin{figure}
\centering
\includegraphics[scale=0.9]{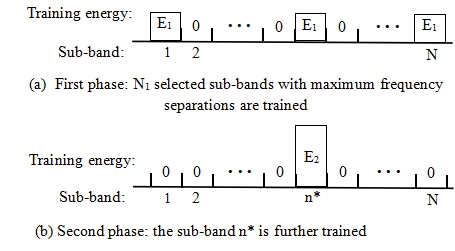}
\caption{Two-phase channel  training for multi-antenna multi-band wireless energy transfer.}\label{F:twoPhaseTraining}
\end{figure}

\section{Problem Formulation}
\subsection{Two-Phase Training}
%We assume that the average power gain $\beta$ is perfectly known at the ET, since it typically varies slowly with time and hence is relatively easy to be estimated. On the other hand, the instantaneous channel realizations $\{\mathbf h_n\}_{n=1}^N$ need to be estimated at the ET based on the proposed two-phase training scheme that exploits the channel reciprocity property.
\subsubsection{Training Phase \rom{1}}
Denote by $\mathcal{N}_1\subset \{1,\cdots, N\}$ with $|\mathcal{N}_1|=N_1$ the $N_1$ selected sub-bands trained in phase \rom{1}.  To maximize the frequency-diversity gain, the sub-bands with the maximum frequency separations are selected  in $\mathcal{N}_1$ so that their channels are most likely to be independent (see Fig.~\ref{F:twoPhaseTraining}(a)), e.g., if $N_1=2$, we have $\mathcal{N}_1=\{1,N\}$.  The received training signals at the ET can be written as
\begin{align}
\mathbf r_n^{\text{\rom{1}}}(t)=\sqrt{E_1} \mathbf h_n \phi_n(t) + \mathbf w_n^{\text{\rom{1}}}(t), \ 0\leq t \leq \tau_1, \ n\in  \mathcal{N}_1,
\end{align}
where %$\mathbf r_n^{\text{\rom{1}}}(t)\in \mathbb{C}^{M\times 1}$ represents the received training signal at the ET;
$E_1$ denotes the training energy used by the ER for each trained sub-band; $\phi_n(t)$ represents the training waveform for sub-band $n$ with normalized energy, i.e., $\int_0^{\tau_1} |\phi_n(t)|^2dt=1$, $\forall n$;  and $\mathbf w_n^{\text{\rom{1}}}(t)\in \mathbb{C}^{M\times 1}$ represents the additive white Gaussian noise received at the ET with power spectrum density $N_0$.  The total energy consumed at the ER for channel training in this phase  is
\begin{align}
E_{\text{tr}}^{\text{\rom{1}}}=\sum_{n\in \mathcal{N}_1} \int_0^{\tau_1}\left|\sqrt{E_1} \phi_n(t)\right|^2=E_1 N_1.
\end{align}

At the ET, the received training signal is first separated over different selected sub-bands; then each $\mathbf r_n^{\text{\rom{1}}}(t)$ passes through a matched filter to get
\begin{align}
\mathbf y_n^{\text{\rom{1}}}= \int_0^{\tau_1} \mathbf r_n^{\text{\rom{1}}}(t) \phi_n^*(t) dt=\sqrt{E_1} \mathbf h_n + \mathbf z_n^{\text{\rom{1}}}, \ n\in \mathcal{N}_1, \label{eq:ynI}
\end{align}
where $\mathbf z_n^{\text{\rom{1}}}\sim \mathcal{CN}(\mathbf 0, N_0 \mathbf I_M)$ denotes the i.i.d. additive Gaussian noise vector.
 Based on \eqref{eq:ynI}, the ET determines the sub-band $n^{\star}$ that has the largest received energy as
\begin{align}
n^{\star}=\arg \underset{n\in \mathcal{N}_1}{\max} \| \mathbf y_n^{\text{\rom{1}}} \|^2. \label{eq:nstar}
\end{align}
The ET then sends  the index $n^{\star}$ to the ER.

\subsubsection{Training Phase \rom{2}}
In the second phase of $\tau_2$ seconds, additional pilot signal $u(t)$ is transmitted by the ER over sub-band $n^{\star}$ with energy $E_2$. With similar processing as that in phase \rom{1}, the received  signal at the ET over sub-band $n^\star$ is
%The received training signal at the ET is then given by
%\begin{align}
%\mathbf Y_{n^{\star}}^{\text{\rom{2}}} = \sqrt{P_2} \mathbf h_{n^{\star}} \mathbf s^H + \mathbf Z_{n^{\star}}^{\text{\rom{2}}}, \label{eq:ynII}
%\end{align}
%where $\mathbf Y_{n^{\star}}^{\text{\rom{2}}}, \mathbf Z_{n^{\star}}^{\text{\rom{2}}} \in \mathbb{C}^{M\times \tau_2}$ contains the received training signals and noises, respectively. The energy consumed at the ER for sending pilot symbols in phase \rom{2} can be expressed as
%\begin{align}
%E_{\text{tr}}^{\text{\rom{2}}}=\left\|\sqrt{P_2}\mathbf s \right\|^2=E_2,
%\end{align}
%with $E_2\triangleq P_2\tau_2$ denoting the training energy for sub-band $n^{\star}$ in phase \rom{2}.
%
% To estimate the MISO channel of sub-band $n^{\star}$,  the received signal $\mathbf Y_{n^{\star}}^{\text{\rom{2}}}$ is projected to $(1/\sqrt{\tau_2}) \mathbf s$ to get a sufficient statistic for the estimation of $\mathbf h_{n^{\star}}$. The resulting signal can be expressed as
\begin{align}
\mathbf y_{n^{\star}}^{\text{\rom{2}}} = \sqrt{E_2} \mathbf h_{n^{\star}} + \mathbf z_{n^\star}^{\text{\rom{2}}}, \label{eq:ynII}
\end{align}
where $\mathbf z_{n^\star}^{\text{\rom{2}}}\sim \mathcal{CN}(\mathbf 0, N_0 \mathbf I_M)$.
%Based on \eqref{eq:ynII}, linear minimum mean-square error (LMMSE) estimation is performed for sub-band $n^\star$.
The ET then performs the linear minimum mean-square error (LMMSE) based estimation for $\mathbf h_{n^{\star}}$  based on $\mathbf y_{n^{\star}}^{\text{\rom{2}}}$.\footnote{In principle, $\mathbf h_{n^{\star}}$ can be estimated based on both observations $\mathbf y_{n^\star}^{\text{\rom{1}}}$ and $\mathbf y_{n^{\star}}^{\text{\rom{2}}}$. To simplify the processing of multi-band energy detection in phase \rom{1} training, we assume that $\mathbf y_{n^\star}^{\text{\rom{1}}}$ is only used for estimating $\|\mathbf h_{n^\star}\|^2$ while  only $\mathbf y_{n^{\star}}^{\text{\rom{2}}}$ is used for estimating $\mathbf h_{n^{\star}}$.} To obtain the optimal LMMSE estimator, we first provide  the following lemma.

\begin{lemma}\label{lemma:exphnSq}
Given that $\mathbf h_n$ and $\mathbf h_m$ are independent $\forall n,m\in \mathcal{N}_1$ and $n\neq m$, the average  power of the MISO channel $\mathbf h_{n^\star}$ over the selected sub-band $n^\star$ can be expressed as
\begin{align}
R_h(N_1,E_1)\triangleq \xE\left [\left \| \mathbf h_{n^\star} \right\|^2 \right]=\frac{\beta^2 E_1 G(N_1,M)+\beta N_0 M}{\beta E_1 + N_0}, \label{eq:Ehnsq}
\end{align}
where $G(N_1,M)\geq M$ is an increasing function with respect to both $N_1$ and $M$ as defined in \eqref{eq:GN1M}. %, with $G(1,M)=M$, $\forall M$.
%where $G(N_1,M)\geq M$ is an increasing function with respect to both $N_1$ and $M$ given by
%\begin{align}
%G(N_1,M)=\sum_{n=1}^{N_1} \binom{N_1}{n}(-1)^{n+1} c_n, \label{eq:GN1M}
%\end{align}
%with $c_n$ defined by \eqref{eq:an} shown at the top of the next page.
%\begin{figure*}
%\begin{equation}
%\begin{aligned}
%c_n=\sum_{k_0+\cdots k_{M-1}=n}\binom{n}{k_0,\cdots,k_{M-1}} \left( \prod_{m=0}^{M-1} \frac{1}{\left(m!\right)^{k_m}}\right)\left(\sum_{m=0}^{M-1}mk_m \right)! \frac{1}{n^{1+\sum_{m=0}^{M-1}mk_m}}.\label{eq:an}
%\end{aligned}
%\end{equation}
%%\begin{equation}\label{eq:QnetN1E1}
%%\begin{aligned}
%%\Qnet(N_1, E_1)=\begin{cases}
%%\eta T P_f R_h(N_1,E_1) + \frac{N_0 M}{R_h(N_1,E_1)}-E_1N_1-2\sqrt{\eta TP_f(M-1)N_0},  & \text{ if } R_h(N_1,E_1) > \alpha, \\
%%\frac{\eta T P_f R_h(N_1,E_1)}{M}-E_1N_1,
%% & \text{ otherwise}.
%%\end{cases}
%%\end{aligned}
%%\end{equation}
%\hrulefill
%\end{figure*}
%Note that in \eqref{eq:an}, the summation is taken over all sequences of non-negative integer indices $k_0$ to $k_{M-1}$ with the sum equal to $n$, and the coefficients $\dbinom{n}{k_0,\cdots,k_{M-1}}=\frac{n!}{k_0!\cdots k_{M-1}!}$ are known as multinomial coefficients. %, which can be computed as
%%\begin{align}
%%\binom{n}{k_0,\cdots,k_{M-1}}=\frac{n!}{k_0!\cdots k_{M-1}!}.
%%\end{align}
\end{lemma}

\begin{IEEEproof}
%Due to the space limitation, the proof is omitted here and is provided in a longer version of this paper \cite{529}.
Please refer to Appendix~\ref{A:exphnSq}.
\end{IEEEproof}

$R_h(N_1,E_1)$ is the average power of the MISO channel when the ``best'' out of the $N_1$  independent sub-band channels  is selected. It can be easily verified that $R_h(N_1,E_1)$ increases with both $N_1$ and $E_1$, as expected.
\begin{lemma}\label{lemma:LMMSE}
The LMMSE estimator $\hat {\mathbf h}_{n^\star}$ of $\mathbf h_{n^\star}$ based on \eqref{eq:ynII} is given by
\begin{align}
\hat {\mathbf h}_{n^\star}=\frac{\sqrt{E_2}R_h(N_1,E_1)}{E_2R_h(N_1,E_1)+N_0 M}\mathbf y_{n^{\star}}^{\text{\rom{2}}}. \label{eq:hnLMMSE}
\end{align}
Define the channel estimation error as $\tilde{\mathbf h}_{n^{\star}} \triangleq \mathbf h_{n^\star}-\hat {\mathbf h}_{n^\star}$. We also have
\begin{align}
&\xE\left[\|\tilde{\mathbf h}_{n^{\star}}\|^2 \right]=\frac{N_0MR_h(N_1,E_1) }{E_2 R_h(N_1,E_1) + N_0 M}, \label{eq:mse}\\
&\xE\left[\|\hat {\mathbf h}_{n^\star}\|^2 \right]=\frac{E_2 R_h^2(N_1,E_1)}{E_2R_h(N_1,E_1)+N_0M},\label{eq:ehhat}\\
&\xE\left[ \tilde {\mathbf h}_{n^\star}^H\hat {\mathbf h}_{n^\star} \right]=0.\label{eq:ecross}
\end{align}
\end{lemma}
\begin{IEEEproof}
%Please refer to \cite{529}.
Please refer to Appendix~\ref{A:LMMSE}.
\end{IEEEproof}

%Note that the LMMSE estimator and the resulting MSE depend on the training designs in both phase \rom{1} and phase \rom{2}. In the special case with $N_1=1$, i.e., only one sub-band is trained in phase \rom{1}, $\mathbf h_{n^{\star}}$ is CSCG distributed as in \eqref{eq:hn} and the corresponding LMMSE estimator in \eqref{eq:hnLMMSE} is identical to the optimal MMSE estimator.

\subsection{Net Harvested Energy Maximization}
After  the two-phase training, energy beamforming is performed by the ET over sub-band $n^{\star}$ based on the estimated channel $\hat{\mathbf h}_{n^{\star}}$ during the remaining time of $T-\tau_1-\tau_2$ seconds. According to \eqref{eq:optxn}, we set   $\mathbf s_{n^\star}=\sqrt{P_f}\hat{\mathbf h}_{n^\star}/\|\hat{\mathbf h}_{n^\star}\|$.
%s $\{\mathbf s_n\}_{n=1}^N$ can be written as
%\begin{align}
%\mathbf s_n =
% \begin{cases}
%\sqrt{P_f} \frac{\hat{\mathbf h}_n s_n}{\|\hat{\mathbf h}_n\|}, & \text{ if } n=n^{\star}, \\
%\mathbf 0, & \text { otherwise},
%\end{cases}
%\end{align}
%with $s$ denoting an arbitrary unit-power symbol.
The resulting  energy harvested at the ER can be expressed as\footnote{We assume that $T$ is sufficiently large so that $T\gg \tau_1+\tau_2$; as a result, the time overhead for channel training is ignored (but energy cost of channel training remains).} %In fact, for any given training energy $E_1$ and $E_2$, $\tau_1$ and $\tau_2$ should be set to the minimum value ($\approx 1/B_s$) as determined by the bandwidth of each sub-channel.
%}
%made as small as $1$ by appropriately choosing the training powers $P_1$ and $P_2$.}
\begin{align}
\hat Q &= \eta T P_f\frac{\left|\mathbf h_{n^{\star}}^H \hat{\mathbf h}_{n^{\star}}\right|^2}{\|\hat {\mathbf h}_{n^{\star}}\|^2}\\
&=\eta T P_f \Big( \|\hat{\mathbf h}_{n^{\star}} \|^2 + \frac{|\tilde{\mathbf h}_{n^{\star}}^H \hat{\mathbf h}_{n^\star}|^2}{\|\hat {\mathbf h}_{n^{\star}}\|^2}
+\tilde{\mathbf h}_{n^{\star}}^H \hat{\mathbf h}_{n^\star} + \hat{\mathbf h}_{n^\star}^H\tilde{\mathbf h}_{n^{\star}}\Big),\label{eq:hatQ2}
\end{align}
where we have used the identity $\mathbf h_{n^{\star}}=\hat{\mathbf h}_{n^{\star}}+ \tilde{\mathbf h}_{n^\star}$ in \eqref{eq:hatQ2}. The average harvested energy at the ER is then obtained  as
\begin{align}
\bar Q & (N_1,E_1,E_2) = \xE\left[ \hat Q\right] \notag \\ %=\eta T P_f \left( \xE\left[ \|\hat{\mathbf h}_{n^{\star}} \|^2 \right]+ \frac{1}{M}\xE\left[ \|\tilde{\mathbf h}_{n^{\star}} \|^2\right]\right) \label{eq:barQ1}\\
%&=\eta T P_f \frac{R_h\left(P_2\tau_2R_h+N_0\right)}{P_2\tau_2R_h+N_0M},
&=\eta T P_f R_h(N_1,E_1)\Big(1-\frac{(M-1)N_0}{E_2R_h(N_1,E_1)+N_0M}\Big), \label{eq:barQ}
\end{align}
where we have used the results in \eqref{eq:mse}-\eqref{eq:ecross}.

It is observed from \eqref{eq:barQ} that the average harvested energy is given by a difference of two terms. The first term, $\eta T P_f R_h$, is the average harvested energy when energy beamforming is based on the perfect knowledge of $\mathbf h_{n^\star}$, with the best  sub-band $n^\star$ determined  via phase \rom{1} training. %Therefore, $\eta T P_f R_h(N_1,E_1)$ includes the effects of both frequency-diversity gain that has been achieved by  training phase \rom{1}, as well as the maximum energy beamforming gain that can be exploited by estimating the MISO channel $\mathbf h_{n^\star}$ in training phase \rom{2}.
 The second term can be interpreted as the loss in energy beamforming performance due to the error in the  estimated MISO channel $\hat{\mathbf h}_{n^\star}$ in phase \rom{2} training. As $E_2/N_0 \rightarrow \infty$, $\mathbf h_{n^\star}$ can be perfectly estimated and hence the second term in \eqref{eq:barQ} vanishes. %In the absence of training phase \rom{2} ($E_2=0$), we have $\bar Q(N_1,E_1,0)=\eta T P_f R_h(N_1,E_1)/M$, i.e., no energy beamforming gain is achieved; on the other hand, if $N_1=1$, i.e., only one sub-band has been trained in phase \rom{1}, we have $\bar Q(1,E_1,E_2)=\eta T P_f \beta \frac{E_2 \beta M + N_0}{E_2 \beta + N_0}$. In this case, no frequency-diversity gain can be achieved. As a result, for a given total energy for channel training, both training phases are needed in general in order to achieve both frequency-diversity and beamforming gains.

 The \emph{net} average harvested energy at the ER, which is the average harvested energy offset by  that used for sending training signals in the two phases, is  given by %can be expressed as \eqref{eq:Qnet} shown at the top of the next page. % as given by %\footnote{We  ignore the circuit energy consumption associated with the energy harvesting process, since it is assumed to be a constant and is irrelevant from our analysis.}
 %\begin{figure*}
 %\begin{equation}
 %\small
\begin{align}\label{eq:Qnet}
 \Qnet & (N_1,E_1,E_2) =\bar Q(N_1,E_1,E_2)- E_1N_1-E_2.
%=\eta T P_f R_h(N_1,E_1)\left(1-\frac{(M-1)N_0}{E_2R_h(N_1,E_1)+N_0M}\right)-E_1N_1-E_2.
\end{align}
%\end{equation}
%\end{figure*}
 The problem of finding the optimal training design to maximize  $\Qnet$ can be formulated as
\begin{align}
\mathrm{(P1):} \qquad \underset{ E_1\geq 0, E_2\geq 0, N_1}{\max} \quad & \Qnet (N_1,E_1,E_2) \notag \\
\text{subject to} \quad &  N_1\in \{1,\cdots, N\}. \notag
%& E_1\geq 0,\  E_2\geq 0. \notag
\end{align}

\section{Optimal Training Design}\label{sec:optimalDesign}
To find the optimal solution to (P1), we first obtain the optimal training energy $E_2$ with $N_1$ and $E_1$ fixed. By discarding irrelevant terms, the resulting sub-problem can be formulated as
\begin{equation}\label{P:solveE2}
\begin{aligned}
\underset{E_2\geq 0}{\min} & \  \frac{(M-1)N_0\eta T P_f R_h(N_1,E_1)}{E_2R_h(N_1,E_1)+N_0M}+E_2,
\end{aligned}
\end{equation}
which is convex with the optimal solution given by
\begin{align}
E_2^\star(N_1,E_1)=\left[\sqrt{\eta T P_f(M-1)N_0} -\frac{N_0M}{R_h(N_1,E_1)}\right]^+,\label{eq:E2star}
\end{align}
where $[x]^+\triangleq \max\{x,0\}$.
By substituting $E_2^\star(N_1,E_1)$ into \eqref{eq:barQ}, the resulting average net  energy as a function of $N_1$ and $E_1$ can be expressed as %\eqref{eq:QnetN1E1} shown on the top of the next page,
%\begin{figure*}
\begin{equation}\label{eq:QnetN1E1}
\begin{aligned}
\Qnet(N_1, E_1)=\begin{cases}
\eta T P_f R_h(N_1,E_1) + \frac{N_0 M}{R_h(N_1,E_1)}-E_1N_1-2\sqrt{\eta TP_f(M-1)N_0},  & \text{ if } R_h(N_1,E_1) > \alpha, \\
\frac{\eta T P_f R_h(N_1,E_1)}{M}-E_1N_1,
 & \text{ otherwise},
\end{cases}
\end{aligned}
\end{equation}
%\hrulefill
%\end{figure*}
where
%\begin{align}
$\alpha\triangleq \sqrt{N_0}M/\eta T P_f (M-1)$.
%\end{align}

As a result, (P1) reduces to
\begin{equation}\label{P:E1N1}
\begin{aligned}
\underset{E_1\geq 0, N_1}{\max} \quad & \Qnet (N_1,E_1)  \\
\text{subject to} \quad &  N_1\in \{1,\cdots, N\}. %  \\
%& E_1\geq 0.
\end{aligned}
\end{equation}

To find the optimal solution to problem \eqref{P:E1N1}, we first obtain the optimal $E_1$ with $N_1$ fixed by solving
\begin{equation}\label{P:E1}
\begin{aligned}
\underset{E_1\geq 0}{\max} \quad & \Qnet (N_1,E_1).
\end{aligned}
\end{equation}

 It can be obtained from \eqref{eq:Ehnsq} that for any fixed $N_1$, as the training energy $E_1$ varies from $0$ to $\infty$, $R_h(N_1,E_1)$ monotonically increases from $\beta M$ to $\beta G(N_1,M)$, i.e.,
\begin{align}
\beta M \leq R_h(N_1,E_1) \leq \beta G(N_1,M), \ \forall E_1\geq 0. \label{eq:Rhbound}
\end{align}
 As a result, problem \eqref{P:E1} can be solved by separately considering the following three cases:

{\it Case 1: $\alpha \geq \beta G (N_1,M)$:}  In this case, we have $R_h(N_1,E_1)\leq \alpha$ and hence $\Qnet(N_1,E_1)=\eta T P_f R_h(N_1,E_1)/M-E_1N_1$, $\forall E_1\geq 0$. %As a result,  problem \eqref{P:E1} reduces to  %problem \eqref{P:phase1OnlyFixN1} as discussed in Section~\ref{sec:phaseI}.
By substituting $R_h(N_1,E_1)$ with  \eqref{eq:Ehnsq}, problem \eqref{P:E1} reduces to
\begin{equation}\label{P:E1case1}
\begin{aligned}
\underset{E_1\geq 0}{\max}\  \eta T P_f \beta  \frac{\beta E_1 G(N_1,M)/M+ N_0}{\beta E_1 + N_0}-E_1N_1,
\end{aligned}
\end{equation}
which is convex with the optimal solution given by
\begin{align}
E_1^{\star}(N_1)=\left[\sqrt{\frac{\eta T P_f N_0\left(G(N_1,M)/M-1 \right)}{N_1}} -\frac{N_0}{\beta}\right]^+.\notag
\end{align}

{\it Case 2: $\alpha \leq \beta M$:} In this case, $R_h(N_1,E_1)  > \alpha$, $\forall E_1\geq 0$. Therefore, $\Qnet(N_1,E_1)$ is given by the first expression of \eqref{eq:QnetN1E1}. After discarding irrelevant terms, problem \eqref{P:E1} can be explicitly written as %\eqref{P:E1case2} shown at the top of the next page.
%\begin{figure*}
\begin{equation}\label{P:E1case2}
%\small
\begin{aligned}
\underset{E_1\geq 0}{\max}\  \eta T P_f  \frac{\beta^2 E_1 G(N_1,M)+\beta N_0 M}{\beta E_1 + N_0}+ \frac{N_0 M (\beta E_1 + N_0)}{\beta^2 E_1 G(N_1,M)+\beta N_0 M}-E_1N_1.
\end{aligned}
\end{equation}
%\vspace{-2ex}
%\end{figure*}
Problem \eqref{P:E1case2} is non-convex in general. However, as the objective function is continuously differentiable, the optimal solution is  given either by $E_1=0$, or by one of the positive stationary points satisfying $\frac{\partial \Qnet(N_1,E_1)}{\partial E_1}=0$, which can be easily determined by solving a quartic equation.

{\it Case 3: $\beta M <  \alpha < \beta G (N_1,M)$:} In this case, it can be obtained that $\Qnet(N_1,E_1)$ in \eqref{eq:QnetN1E1} can be explicitly expressed as \eqref{eq:N1E1Case3} shown at the top of the next page,
%\begin{figure*}
%\begin{equation}
%\small
\begin{align}\label{eq:N1E1Case3}
\Qnet(N_1,E_1)=
\begin{cases}
\eta T P_f  \frac{\beta^2 E_1 G(N_1,M)/M+\beta N_0}{\beta E_1 + N_0}-E_1N_1, & \text{ if } E_1\leq E_0,\\
\eta T P_f  \frac{\beta^2 E_1 G(N_1,M)+\beta N_0 M}{\beta E_1 + N_0}+ \frac{N_0 M (\beta E_1 + N_0)}{\beta^2 E_1 G(N_1,M)+\beta N_0 M}-E_1N_1 -2\sqrt{\eta TP_f(M-1)N_0},& \text{ otherwise},
\end{cases}
\end{align}
%\end{equation}
%\hrulefill
%\end{figure*}
where
%\begin{align}
$E_0\triangleq \frac{N_0(\alpha-\beta M)}{\beta(\beta G-\alpha)}$.
%\end{align}
Similar to that in Case 2, the optimal solution to  problem \eqref{P:E1} with $ \Qnet (N_1,E_1)$ given in \eqref{eq:N1E1Case3} is given either by the boundary point $E_1=0$  or one of  the stationary points, which can be readily determined by solving a quartic equation.

With problem \eqref{P:E1} solved for all three cases as discussed above, the corresponding optimal value $\Qnet^{\star}(N_1)$ as a function of $N_1$ can be readily determined. Therefore, finding the optimal solution to problem \eqref{P:E1N1} and that to the original problem (P1) reduces to determining the optimal number of sub-bands to be trained, i.e., $N_1^{\star}=\arg\underset{1\leq N_1 \leq N}{\max} \Qnet^{\star}(N_1)$, which can be easily found by exhaustive search.

\begin{figure}
\centering
\includegraphics[scale=0.3]{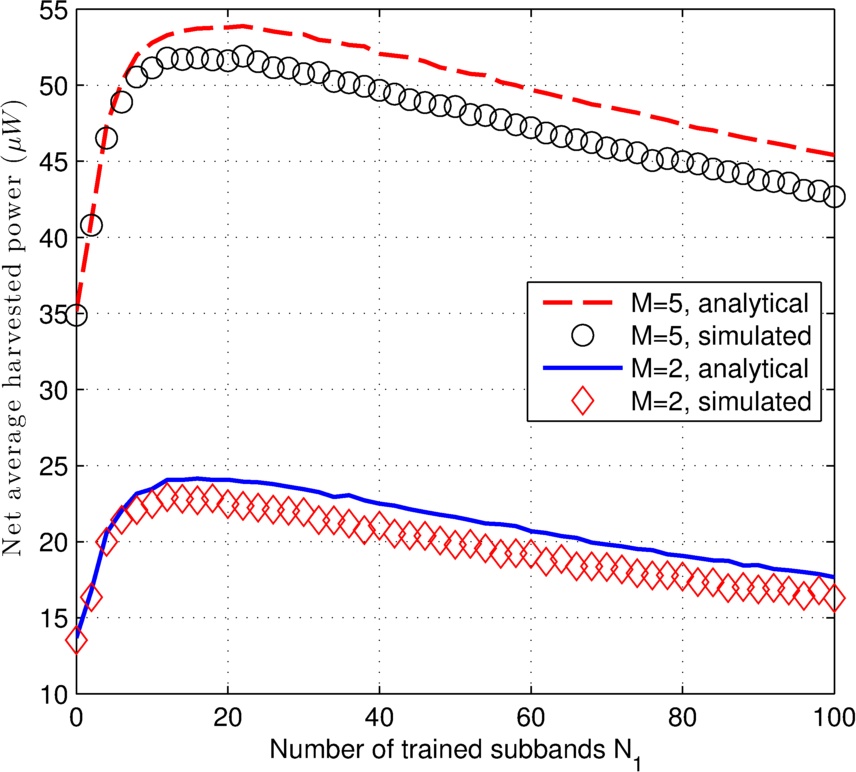}
\caption{Net average harvested power versus the number of trained sub-bands $N_1$. %, $B=6$MHz, $B_s=150$ kHz, $\sigma_{rms}=1 \mu$s, $B_c=160$ kHz, $P_f=37$dBm, $\beta=-60dB$, $N_0=-160dbm/Hz$, $T=50$ms.
}\label{F:QnetVsN1}
\end{figure}

\section{Numerical Results}
%In this section, numerical examples are provided to corroborate our study. We assume that the ET has the maximum transmit power of $30$ dBm, i.e., $P_f=1$ Watt. The average signal attenuation from the ET to the ER is assumed to be $40$dB, i.e., $\beta=10^{-4}$. Furthermore, the training noise power is assumed be $N_0=-50$ dBm and the energy harvesting efficiency is $\eta=0.5$. %The resulting ETNR can then be calculated as $\Gamma=0.5$.
In this section, numerical examples are provided to corroborate our study.  To model the frequency-selective channel,
 we assume a multi-path power delay profile with the exponential distribution  $A(\tau)=\frac{1}{\sigma_{\text{rms}}} e^{-\tau/\sigma_{\text{rms}}}$, $\tau\geq 0$, where $\sigma_{\text{rms}}$  denotes the root-mean-square (rms) delay spread. We set $\sigma_{\text{rms}}=1\mu$s so that the $50\%$ channel coherence bandwidth, i.e., the frequency separation for which the amplitude correlation is $0.5$,  is $B_c=\frac{1}{2\pi \sigma_{\text{rms}}}\approx 160$ kHz. The total available spectrum for energy transmission is $B=10$MHz, which is divided into $N=100$ sub-bands each with bandwidth $B_s=100$kHz.  The average power attenuation between the ET and the ER is assumed to be $50$ dB, i.e., $\beta=10^{-5}$, and the transmission power at the ET is set as $P_f=1$watt or $30$dBm. The power spectrum density of the training noise received at the ET is $N_0=-120$dBm/Hz.  The energy harvesting efficiency at the ER is set as $\eta=0.8$.

In Fig.~\ref{F:QnetVsN1}, by varying the number of sub-bands $N_1$ that are trained in phase \rom{1}, the net average harvested power achieved by the proposed two-phase training scheme is plotted for $M=5$ and $M=2$, where the average is taken over $10000$ random channel realizations. The channel block length is set as $T=0.5$ms. %In order to maximize the diversity gain, for any given $N_1$, the sub-bands with the maximum frequency separations are trained in phase \rom{1} (see Fig.\ref{F:twoPhaseTraining}), e.g., if $N_1=2$, we have $\mathcal{N}_1=\{1,N\}$.
 The analytical result obtained in Section~\ref{sec:optimalDesign}, i.e.,  $\Qnet^{\star}(N_1)/T$ with $\Qnet^{\star}(N_1)$ denoting the optimal value of problem \eqref{P:E1},  is also shown in Fig.~\ref{F:QnetVsN1}. It is observed that the simulation and analytical results match well for small and moderate $N_1$ values, for which the assumption of independent channels between any two  sub-bands as in Lemma~\ref{lemma:exphnSq} is more valid. %As $N_1$ increases, the trained sub-bands become correlated. In this case, Fig.~\ref{F:QnetVsN1} shows that the analytical result derived in this paper provides a reasonable approximation for the simulated values. % and hence the theoretical results only give a performance upper bound. Fig.~\ref{F:QnetVsN1} shows that for properly designed sub-channel bandwidth $B_s$, the theoretical result derived in this paper provides a reasonable approximation for the actual harvested energy for correlated channels.
 Furthermore, Fig.~\ref{F:QnetVsN1} shows that there is an optimal number of sub-bands trained to maximize the net harvested energy, as a result of the trade-off between achieving more frequency-diversity gain (with lager $N_1$) and reducing the training energy ($E_1N_1$ in phase \rom{1}).

%\begin{figure}
%\centering
%\includegraphics[scale=0.3]{N1OptVsT}
%\caption{Optimal number of training sub-bands $N_1^\star$ versus block length $T$ for $M=2$ and $M=5$.}\label{F:N1OptVsT}
%\end{figure}

%In Fig.~\ref{F:N1OptVsT}, the optimal number of sub-bands $N_1^{\star}$ that should be trained is plotted against the channel block length $T$ for $M=2$ and $M=5$. It is observed that for both $M$ values,  $N_1^{\star}$ increases with $T$. This is expected since as the channel changes more slowly,  training is more beneficial and hence more sub-bands should be trained to achieve larger frequency-diversity gain. As $T$ becomes sufficiently large, all the available sub-bands should be trained. It is also observed from Fig.~\ref{F:N1OptVsT} that more sub-bands should be trained for $M=5$ than that for $M=2$, since it is affordable to have more channel training when the energy beamforming gain is larger.

 \begin{figure}
\centering
\includegraphics[scale=0.3]{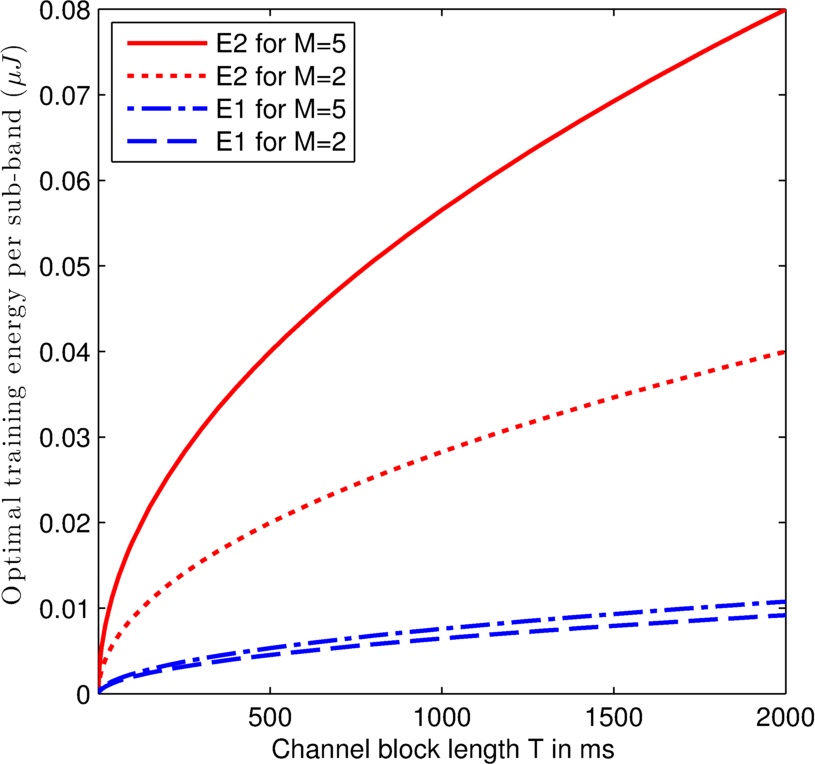}
\caption{Optimal training energy $E_1$ and $E_2$ versus block length $T$ for $M=2$ and $M=5$.}\label{F:P1P2VsT}
\end{figure}

 In Fig.~\ref{F:P1P2VsT}, the optimal training energy per sub-band $E_1$ and $E_2$ in phases \rom{1} and \rom{2}, respectively, are plotted against the channel block length $T$, with $T$ ranging from $0.1$ms to $2$ seconds. It is observed that $E_1$ and $E_2$ both increase with $T$, as expected. % since as the block length increases, training is more beneficial and hence more training energy can be afforded to achieve larger frequency-diversity and energy-beamforming gains.
  Furthermore, for both setups, $E_2$ is significantly larger than $E_1$, since in phase \rom{2}, only the selected sub-band needs to be further trained, whereas the training energy in phase \rom{1} needs to be distributed over  $N_1$ sub-bands to exploit the frequency-diversity.

  %. This is due to the fact that phase \rom{2} only needs to train the particular sub-band that has been determined in phase \rom{1} to achieve energy beamforming gain, whereas phase \rom{1} needs to train $N_1$ sub-bands in order to exploit the frequency-diversity gain. As a result, more training energy per sub-band can be afforded in phase \rom{2}  than that in phase \rom{1}. %Furthermore, similar to Fig.~\ref{F:N1OptVsT}, Fig.~\ref{F:P1P2VsT} also shows that more training should be applied when the ET has more antennas, and hence potentially higher beamforming gains can be exploited by channel training.

In Fig.~\ref{F:PnetVsTN100M5}, the net average harvested power based on the proposed two-phase training scheme is plotted against block length $T$ with $M=5$. The following four benchmark schemes are also included for comparison: i) perfect CSIT, whose average harvested energy can be obtained as $\bar{Q}_{\max}=\eta T  P_f \beta G(N,M)$; ii) no CSIT, with $\bar{Q}_{\text{noCSIT}}=\eta T P_f \beta$; iii) phase \rom{1} training only, which corresponds to the special case of the two-phase training scheme with $E_2=0$; iv) phase \rom{2} training only, which corresponds to the two-phase scheme with $E_1=0$.  It is observed from Fig.~\ref{F:PnetVsTN100M5} that the proposed two-phase training  scheme  approaches to the performance upper bound with perfect CSIT as $T$ increases, and significantly outperforms the other  three benchmark schemes. It is also worth noting that for multi-antenna frequency-selective WET systems, exploiting either frequency-diversity gain or beamforming gain alone is far from optimal; instead, a good balance between these  two gains as achieved in the proposed two-phase training optimization  is needed.
 \begin{figure}
\centering
\includegraphics[scale=0.3]{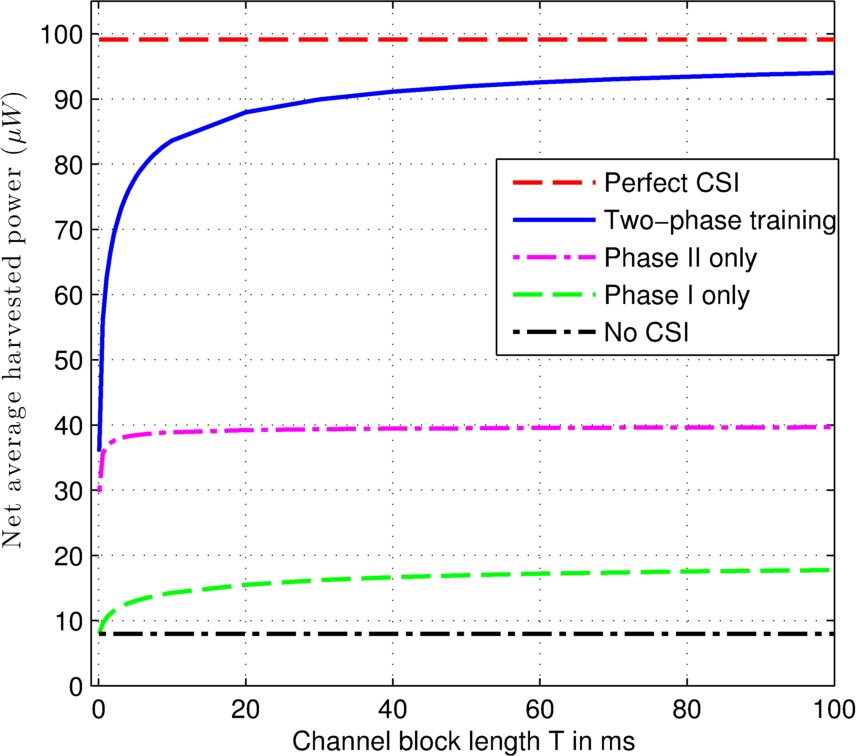}
\caption{Net average harvested power with $M=5$.}\label{F:PnetVsTN100M5}
\end{figure}

 \section{Conclusion}
 This paper studies the optimal  training design for a MISO  WET system in frequency-selective channels. By exploiting channel reciprocity, a two-phase training scheme is proposed to exploit the frequency-diversity and energy-beamforming gains, respectively. A closed-form expression has been derived for the average harvested energy. The optimal training scheme, including the number of independent sub-bands trained and the energy allocated for each of the two training  phases, is derived. Numerical results are provided to validate our analysis and  show the effectiveness of the proposed scheme by optimally balancing the achievable diversity and beamforming gains with limited training energy.

%\begin{figure}
%\centering
%\includegraphics[scale=0.3]{QnetVsM}
%\caption{Net average harvested energy normalized by the number of ET antennas $M$, with $N=20$, $T=200$.}\label{F:QnetVsM}
%\end{figure}
%
%\begin{figure}
%\centering
%\includegraphics[scale=0.3]{QnetVsNTotal}
%\caption{Net average harvested energy versus the number of sub-bands $N$, with $M=2$, $T=100$.}\label{F:QnetVsNTotal}
%\end{figure}

%\section{Conclusion}
%Under the point-to-point and Rayleigh fading  setups, this  paper studies the optimal design of channel training for MIMO WET systems. Assuming channel reciprocity, the forward link channel is efficiently estimated at the ET based on the training signals sent by the ER in the reverse link. We derive the optimal training strategy to maximize the net average  energy harvested at the ER by taking into account the energy consumed for channel training. In  future work, we will extend the results in this paper to more general setups with more than one ERs as well as other practical fading channel models.

%\onecolumn

\appendices
\section{A Useful Lemma}
\begin{lemma}\label{lemma:expV}
Let $\mathbf v_1,\cdots, \mathbf v_{N_1}\in \mathbb{C}^{M\times 1}$ be $N_1$ i.i.d. zero-mean CSCG random vectors distributed as $\mathbf v_n\sim \mathcal{CN}(\mathbf 0, \sigma^2_v \mathbf I_M)$, $\forall n$. Then we have
\begin{align}
\xE\left[\underset{n=1,\cdots, N_1}{\max}\left\| \mathbf v_n \right\|^2 \right]=\sigma^2_v G(N_1,M), \label{eq:expMax}
\end{align}
where $G(N_1,M)$ is a function of $N_1$ and $M$ given by
\begin{align}
G(N_1,M)=\sum_{n=1}^{N_1} \binom{N_1}{n}(-1)^{n+1} c_n, \label{eq:GN1M}
\end{align}
with
\begin{align}
c_n=\sum_{k_0+\cdots k_{M-1}=n}\binom{n}{k_0,\cdots,k_{M-1}} \left( \prod_{m=0}^{M-1} \frac{1}{\left(m!\right)^{k_m}}\right)\left(\sum_{m=0}^{M-1}mk_m \right)! \frac{1}{n^{1+\sum_{m=0}^{M-1}mk_m}}.\label{eq:an}
\end{align}
Note that in \eqref{eq:an}, the summation is taken over all sequences of non-negative integer indices $k_0$ to $k_{M-1}$ with the sum equal to $n$, and the coefficients $\dbinom{n}{k_0,\cdots,k_{M-1}}$ are known as multinomial coefficients, which can be computed as
\begin{align}
\binom{n}{k_0,\cdots,k_{M-1}}=\frac{n!}{k_0!\cdots k_{M-1}!}.
\end{align}
\end{lemma}

\begin{proof}
Define the random variables $V_n\triangleq\|\mathbf v_n\|^2$, $n=1,\cdots, N_1$. It then follows that $V_1,\cdots, V_{N_1}$ are i.i.d. Erlang distributed with shape parameter $M$ and rate $\lambda=1/\sigma^2_v$, whose cumulative distribution function (CDF) is given by
\begin{align}
F_{V_n}(v)=\pr (V_n\leq v)=1-\sum_{m=0}^{M-1} \frac{1}{m!} e^{-\lambda v}(\lambda v)^m, \ \forall n.
\end{align}
Let $V\triangleq \underset{n=1,\cdots, N_1}{\max} V_n$. Then the CDF of $V$ can be obtained as
\begin{align}
F_V(v)=\pr \left(v_1\leq v,\cdots, v_{N_1} \leq v \right)=\prod_{n=1}^{N_1}F_{V_n}(v)=\left(1-\sum_{m=0}^{M-1} \frac{1}{m!} e^{-\lambda v}(\lambda v)^m \right)^{N_1}.
\end{align}
With binomial expansion, the expectation of $V$ can be  expressed as
\begin{align}
\xE\left[ V\right]&=\int_0^\infty \left(1-F_V(v)\right)dv\\
&=\sum_{n=1}^{N_1} \binom{N_1}{n} (-1)^{n+1} a_n,\label{eq:Ev}
\end{align}
where
\begin{align}
a_n & =  \int_0^{\infty} e^{-\lambda nv}\left(\sum_{m=0}^{M-1} \frac{1}{m!}(\lambda v)^m  \right)^n dv\\
&=\sum_{k_0+\cdots k_{M-1}=n}\binom{n}{k_0,\cdots, k_{M-1}}
\left( \prod_{m=0}^{M-1} \frac{1}{\left(m!\right)^{k_m}}\right)\lambda^{\sum_{m=0}^{M-1}mk_m } \int_0^\infty e^{-\lambda nv} v^{\sum_{m=0}^{M-1} mk_m}dv \label{eq:multinom}\\
&=\frac{1}{\lambda} c_n=\sigma_v^2 c_n, \label{eq:intIdentity}
\end{align}
where \eqref{eq:multinom} follows from the multinomial  expansion theorem, and  \eqref{eq:intIdentity} follows from the integral identity $\int_0^\infty x^n e^{-\mu x}dx=n! \mu^{-n-1}$(\cite{524}3.351).
The result in \eqref{eq:expMax} can then be obtained by substituting \eqref{eq:intIdentity} into \eqref{eq:Ev}.

Furthermore, it can be directly obtained from \eqref{eq:expMax} that $G(N_1,M)$ is an increasing function with respect to both $N_1$ and $M$, with $G(1,M)=M$, $\forall M$.

This completes the proof of Lemma~\ref{lemma:expV}.
\end{proof}

\section{Proof of Lemma~\ref{lemma:exphnSq}}\label{A:exphnSq}
Note that in the absence of training phase \rom{1} or only one sub-band is trained ($N_1=1$), the distribution of $\mathbf h_{n^\star}$ is simply given by \eqref{eq:hn}. In this case, $\xE\left [\left \| \mathbf h_{n^\star} \right\|^2 \right]=\beta M$, which is equal to that obtained by evaluating \eqref{eq:Ehnsq} with $E_1=0$ or $N_1=1$. For the general scenario with $N_1\geq 2$, $n^{\star}$ is determined by the sub-band with the maximum total received energy as in \eqref{eq:nstar}. As a consequence, the corresponding channel vector $\mathbf h_{n^\star}$ statistically depends on all the $N_1$ channels $\mathbf h_1,\cdots, \mathbf h_{N_1}$ via \eqref{eq:ynI} and \eqref{eq:nstar}. To exploit such a relationship, we first show the following result:
\begin{lemma}\label{lemma:statEqv}
The input-output relationship in \eqref{eq:ynI} is statistically equivalent to
\begin{align}\label{eq:hnEqv}
\mathbf h_n = \frac{\beta \sqrt{E_1}}{\beta E_1 + N_0} \mathbf y_n^{\text{\rom{1}}}+ \sqrt{\frac{\beta N_0}{\beta E_1 + N_0}}\mathbf t_n, \ n=1,\cdots, N_1,
\end{align}
where $\mathbf t_n\sim \mathcal{CN}(\mathbf 0, \mathbf I_M)$ is a CSCG random vector independent of $\mathbf y_n^{\text{\rom{1}}}$, i.e.,
\begin{align}
\xE\left [\mathbf  y_n^{\text{\rom{1}}} \mathbf t_n^H\right]=\mathbf 0, \ n=1,\cdots, N_1.
\end{align}
\end{lemma}

\begin{IEEEproof}
It follows from \eqref{eq:hn} and \eqref{eq:ynI} that $\mathbf y_n^{\text{\rom{1}}}$ is a CSCG random vector distributed as
\begin{align}\label{eq:ynIStat}
\mathbf y_n^{\text{\rom{1}}} \sim \mathcal{CN}\left(\mathbf 0, (\beta E_1+N_0) \mathbf I_M\right), \ \forall n.
\end{align}
Furthermore, the cross-correlation between $\mathbf h_n$ and $\mathbf y_n^{\text{\rom{1}}}$ is
\begin{align}\label{eq:crossCorr}
\xE\left[\mathbf y_n^{\text{\rom{1}}} \mathbf h_n^H \right]=\beta \sqrt{E_1} \mathbf I_M.
\end{align}

To prove Lemma~\ref{lemma:statEqv}, it is sufficient to show that the random vector $\mathbf h_n$ obtained by \eqref{eq:hnEqv} has the same distribution as \eqref{eq:hn}, and also has the same cross-correlation with $\mathbf y_n^{\text{\rom{1}}}$ as \eqref{eq:crossCorr}. The desired results can be easily verified based on \eqref{eq:hnEqv} and \eqref{eq:ynIStat}.
\end{IEEEproof}

By applying Lemma~\ref{lemma:statEqv}, we can obtain the following result
\begin{align}
\xE\left[ \left\| \mathbf h_{n^\star} \right\|^2 \right]&=\frac{\beta^2 E_1}{(\beta E_1 + N_0)^2}\xE\left[ \left \| \mathbf y_{n^\star}^{\text{\rom{1}}} \right\|^2\right]+\frac{\beta N_0 M}{\beta E_1 + N_0}\\
&=\frac{\beta^2 E_1}{(\beta E_1 + N_0)^2}\xE\left[ \underset{n=1,\cdots, N_1}{\max}\left \| \mathbf y_{n}^{\text{\rom{1}}} \right\|^2\right]+\frac{\beta N_0 M}{\beta E_1 + N_0} \label{eq:usenstar}\\
&=\frac{\beta^2 E_1 G(N_1,M)+\beta N_0 M}{\beta E_1 +N_0},\label{eq:Ehnsq2}
\end{align}
where \eqref{eq:usenstar} follows from \eqref{eq:nstar}, and \eqref{eq:Ehnsq2} is true due to Lemma~\ref{lemma:expV} and \eqref{eq:ynIStat}.

This completes the proof of Lemma~\ref{lemma:exphnSq}.

\section{Proof of Lemma~\ref{lemma:LMMSE}}\label{A:LMMSE}
Since both $\mathbf h_{n^\star}$ and $y_{n^{\star}}^{\text{\rom{2}}}$ are zero-mean random vectors with i.i.d entries, the LMMSE estimator can be expressed as $\hat {\mathbf h}_{n^\star}=b\mathbf y_{n^{\star}}^{\text{\rom{2}}}$, with $b$ a complex-valued parameter to be determined.  % Then the estimation error $\tilde{\mathbf h}_{n^{\star}}$ can be expressed as
%\begin{align}
%\tilde{\mathbf h}_{n^{\star}}& \triangleq \mathbf h_{n^\star}-\hat {\mathbf h}_{n^\star}\\
%&=\left(1-b\sqrt{E_2} \right) \mathbf h_{n^\star} -b \mathbf z_{n^{\star}}^{\text{\rom{2}}}.
%\end{align}
The corresponding MSE can be expressed as
\begin{align}
e&=\xE\left[ \left\| \tilde{\mathbf h}_{n^{\star}}\right\|^2\right] =\xE\left[ \left\| \left(1-b\sqrt{E_2} \right) \mathbf h_{n^\star} -b \mathbf z_{n^{\star}}^{\text{\rom{2}}} \right\|^2\right] \\
&=\left|1-b\sqrt{E_2} \right|^2 R_h(N_1,E_1)  + |b|^2 N_0 M\\
&=|b|^2 \left( E_2 R_h(N_1,E_1) + N_0 M\right)-(b+b^*)\sqrt{E_2} R_h(N_1,E_1)+R_h(N_1,E_1).
\end{align}
By setting the derivative of $e$ with respect to $b^*$ equals to zero, the optimal coefficient $b$  can be obtained as
\begin{align}
b=\frac{\sqrt{E_2}R_h(N_1,E_1)}{E_2R_h(N_1,E_1)+N_0M}.
\end{align}
The resulting MMSE can be obtained accordingly.

Furthermore, the following result can be obtained
\begin{align}
\xE\left[\|\hat {\mathbf h}_{n^\star}\|^2 \right]=|b|^2 \xE\left[ \|\mathbf y_{n^{\star}}^{\text{\rom{2}}}\|^2\right]=\frac{E_2 R_h^2(N_1,E_1)}{E_2R_h(N_1,E_1)+N_0M}.
\end{align}

To show that $\xE\left[ \tilde {\mathbf h}_{n^\star}^H\hat {\mathbf h}_{n^\star} \right]=0$, we will use the following result
\begin{align}
\xE\left[ {\mathbf h}_{n^\star}^H \hat {\mathbf h}_{n^\star}\right]=b \xE\left[ {\mathbf h}_{n^\star}^H \mathbf y_{n^{\star}}^{\text{\rom{2}}}\right]=\frac{E_2 R_h^2(N_1,E_1)}{E_2R_h(N_1,E_1)+N_0M}=\xE\left[\|\hat {\mathbf h}_{n^\star}\|^2 \right].
\end{align}
Therefore, we have
\begin{align}
\xE\left[ \tilde {\mathbf h}_{n^\star}^H\hat {\mathbf h}_{n^\star} \right]=
\xE\left[ {\mathbf h}_{n^\star}^H \hat {\mathbf h}_{n^\star}\right]-\xE\left[\|\hat {\mathbf h}_{n^\star}\|^2 \right]=0,
\end{align}
where we have used the identity $\tilde {\mathbf h}_{n^\star}={\mathbf h}_{n^\star}-\hat {\mathbf h}_{n^\star}$.

This completes the proof of Lemma~\ref{lemma:LMMSE}.

\bibliographystyle{IEEEtran}
\bibliography{IEEEabrv,IEEEfull}

\end{document}